\documentclass[12pt]{article}
\usepackage{fullpage}
\pdfoutput=1

\usepackage{amsthm}
\usepackage{amsfonts}
\usepackage{amsmath}
\usepackage{enumerate}
\usepackage{setspace}
\usepackage{wasysym}
\usepackage{gastex}
\usepackage{soul}
\usepackage{amssymb}
\usepackage{cite}
\usepackage{algorithmic}
\usepackage[ruled,linesnumbered]{algorithm2e}
\usepackage[pdftex]{graphicx,color} 
\usepackage[table,xcdraw]{xcolor}
\usepackage[section]{placeins} 

\usepackage{authblk}
\usepackage{hhline}

\usepackage[final]{changes}		

\newcommand{\bea}{\begin{eqnarray}}
\newcommand{\eea}{\end{eqnarray}}

\newtheorem{theorem}{Theorem}
\newtheorem{lemma}[theorem]{Lemma}

\newtheorem{dfn}[theorem]{Definition}
\newtheorem{claim}[theorem]{Claim}

\newtheorem{obs}[theorem]{Observation}

\DeclareMathOperator{\diam}{diam}

\begin{document}
	
\title{Fast approximation algorithms for $p$-centres in large $\delta$-hyperbolic graphs
}

\author{Katherine Edwards\thanks{Department of Computer Science, Princeton University, Princeton, NJ 08540}, W. Sean Kennedy\thanks{Mathematics of Networks and Communications Department, Bell Labs, Nokia, NJ 07974}, Iraj Saniee$^{\dagger}$
}
\affil{\texttt{\texttt{ke@princeton.edu}, \{kennedy,iis\}@research.bell-labs.com}}
\date{}

\maketitle
\def\thepage {} 
\thispagestyle{empty}
\pagenumbering{arabic}

\begin{abstract}
We provide a quasilinear time algorithm for the $p$-center problem with an additive error less than or equal to 3 times the input graph's hyperbolic constant.  
Specifically, for the graph $G=(V,E)$ with $n$ vertices, $m$ edges and hyperbolic constant $\delta$, we construct an algorithm for $p$-centers in time $O(p(\delta+1)(n+m)\log(n))$ with radius not exceeding $r_p + \delta$ when $p \leq 2$ and $r_p + 3\delta$ when $p \geq 3$,  where $r_p$ are the optimal radii. 
Prior work identified $p$-centers with accuracy $r_p+\delta$ but with time complexity $O((n^3\log n + n^2m)\log(\diam(G)))$ which is impractical for large graphs.
\end{abstract}

\section{Introduction}
The $p$-center algorithm is a discrete variant of arguably one of the most frequently used clustering algorithms, the $k$-means clustering.  
The goal of the $p$-center algorithm is to identify on a given graph a pre-specified number $p$ of vertices or centers, such that the maximum distance of any graph vertex to its nearest $p$-center is minimized.    
For any given $p$, the algorithm naturally partitions a graph into $p$ clusters induced by the position of its $p$-centers.   
Clusters induced by the $p$-centers are not necessarily balanced as these are determined strictly by the metric properties of the graph.  
Thus $p$-center clustering is more appropriate for distance-based partitioning or classification than other frameworks, such as community detection.     
Unfortunately, as a clustering algorithm the complexity of the $p$-center algorithm is generally prohibitive, $O(n^p$) for an $n$-node graph, making it inapplicable to even moderate size graphs.

Proved nearly four decades ago, Shier's minimax result for trees and metric trees~\cite{SH77} leads to an exact algorithm with quasilinear time complexity (in the number of vertices  and edges of the graph) for determination of an optimal set of $p$-centers by repeatedly finding diagonal pairs on the graph and carving out a ball containing one end of the current diagonal pair. 
Hochbaum and Shmoys~\cite{HS85} give a (multiplicative) 2-approximation algorithm for determining $p$-centres in graphs satisfying the triangle inequality with running time $O(m \log_2 m)$.
Subsequently, Dyer and Frieze~\cite{DF85} improve this to a 2-approximation algorithm with running time $O(np)$.
These algorithms are, in a sense, best possible as Hsu and Nemhauser~\cite{HN79} show that determining an $\alpha$-approximate solution to $p$-centers is NP-hard whenever $\alpha < 2$.

In an insightful paper \cite{ChEs07}, Chepoi and Estellon essentially apply the technique of Shier \cite{SH77} to graphs with small hyperbolic constant, $\delta$. 
These are graphs whose metric structure differs from the metric structure of a tree by a fixed constant (as explained in Section \ref{sec:defn} and, in particular, Section \ref{sec:hyper} and Figure \ref{fig:triangle}.  For more details see \cite{Gromov87, BrHa99,ChEs07}).  
The algorithmic version of this scheme \cite{ChDrEsHaVa08} gives rise to what is essentially an $O(n^3)$ approximation for $p$-center on an $n$-vertex graph with hyperbolic constant $\delta$ appearing both as a prefactor in the complexity expression and also in the degree of approximation in terms of an additive constant to the radius of the optimal $p$-center partition.  
Of course the polynomial time complexity $O(n^3)$ is still impractical for graphs of hundreds of thousands to millions of nodes as would be even a quadratic complexity approximation.
 
Since there is evidence that real-life networks extracted from social media, co-authorship and collaboration, friendship and many other settings, have small hyperbolic constants \cite{KNS}, it would be desirable to know if the cubic complexity is tight or can be further reduced, at least by negotiating on the degree of the approximation.  
In this paper we show that that by giving up to $3\delta$ in the (additive) approximation, one can achieve a quasilinear time $p$-center approximation. 
As such, this scheme is the first $p$-center approximation applicable to large graphs, particularly when $p$ is relatively small, for example in the range $10-10^4$ and $n$ is large, for example, $10^5-10^9$ vertices.

In the following sections we describe how the cubic complexity of \cite{ChEs07} to quasilinear reduction is achieved without adding more than $3\delta$ to the radius of the optimal $p$-center clusters.
In Section \ref{sec:defn} we outline necessary definitions, in particular, for geodesic metric spaces (Section \ref{sec:geo}) and hyperbolicity (Section \ref{sec:hyper}).  
We then turn to a more formal discussion of $p$-centers, $p$-packings, and the dual problems which take center stage in our discussion (Section \ref{sec:ps}). 
In Section \ref{sec:palgs} we focus on algorithms for solving and approximating these problems on $\delta$-hyperbolic graphs.
The formal statements of our main results are also found in Section \ref{sec:palgs}.
Section \ref{sec:smallp} contains the proofs of the main results.  
We finish in Section \ref{sec:exp} with  experimental validation of our algorithms.

\section{Definitions and notation}
\label{sec:defn}
Let $G = (V,E)$ be an undirected graph, with $V$ the set of vertices and $E$ the set of edges.
To each edge $uv$, we associate a line segment of length $1$, so that we may refer to any point on $uv$ at distance $t$ from $u$ and $1-t$ from $v$ $(0\leq t \leq 1)$.
This (uncountably infinite) set of points of $G$ is denoted $A(G)$.
We will use the notation $n = |V(G)|$ and $m = |E(G)|$.
In this paper, the distance $d(u,v)$ between any two points $u$ and $v$ 
in $A(G)$ is the length of a shortest path between them in $G$.
When $u$ and $v$ are vertices, we write $[u,v]$ to refer to a shortest (also called {\em geodesic}) path. Note that shortest paths need not be unique.
For a geodesic path $P=[u,v]$ and $i\in [0, d(u,v)]$, the point $P[i]$ is the one at distance $i$ from $u$ on $P$.

\subsection{Geodesic metric spaces and graphs}
\label{sec:geo}
Let $(X,d)$ be a metric space.
If $x,y$ are points in $X$, a \emph{geodesic segment} $[x,y]$, when it exists, is a continuous curve parametrized by the line segment $[a,b]$ of length $d = d(x,y)$. 
That is, a map $\rho:[0,d]\rightarrow X$ with $\rho(0) = x, \rho(d) = y$ and $d(\rho(s),\rho(t)) = |s - t|$ for each $s,t \in [0,d]$.
A metric space is \emph{geodesic} if there exists a geodesic segment joining every pair of points. Note that geodesic segments need not be unique, e.g. a diagonal pair of points on a cycle.

Any graph as we have defined above can be viewed as a geodesic metric space $(A(G), d)$.
Such a metric space is called \emph{graphic} and it will be convenient in what follows to think of graphs in this way.
In a graphic metric space, a geodesic $[x,y]$ is simply a shortest path from $x$ to $y$ regardless of $x$ and $y$ being in $V(G)$ or in $A(G)$.

Let $S\subseteq X$ be compact.
The \emph{diameter} $\diam(S)$ of is the maximum length of a geodesic between two vertices in $S$.
For $u\in S$, $F_S(u)$ is the set of points in $S$ whose distance from $u$ is maximum.
Two points $u,v \in S$ are \emph{diametrical} if $d(u,v) = \diam(S)$.
They are \emph{locally diametrical} if $u\in F_S(v)$ and $v\in F_S(u)$. It follows that $d(u,v\in F_S(u)) \leq diam(S)$ and $d(v,u\in F_S(v)) \leq diam(S)$.

If $v$ is a point of $A(G)$ and $r\in \mathbb{R}$, we write $B_r(v)$ for the closed 
ball of radius $r$ about $v$, i.e. all points at distance at most $r$ from $v$.
For a geodesic path $P=[u,v]$ 
and for the length $0 \leq \theta < d(u,v)$, the point $i=[u,v][\theta] \in A(G)$ is at distance $\theta$ from $u$ on $P$. When there is no ambiguity, we identify the point $i=P[\theta]$ with the length $\theta$. Clearly the two points $[u,v][i]$ and $[v,u][i]$ do not generally coincide.

\subsection{Hyperbolicity}
\label{sec:hyper}
The concept of hyperbolicity of a metric space was introduced by Rips and Gromov in \cite{Gromov87}.
There are several essentially equivalent definitions but in this paper we will mainly use the \emph{$\delta$-thin-triangle} characterization.\footnote{For a  comprehensive treatment of $\delta$-hyperbolicity see \cite{BrHa99}.}
For points $x,y,z$ in $X$,
we write $\Delta(x,y,z)$ to denote a \emph{geodesic triangle} formed by $x,y,z$; that is the union of three geodesics $[x,y], [y,z], [x,z]$ (usually the choice of geodesics won't matter).

Given a geodesic triangle $\Delta \equiv \Delta(x,y,z)$, let $\pi$
be half the perimeter, 
$\pi = \tfrac{1}{2} (d(x,y) + d(y,z) + d(x,z))$ and
define $\alpha_x = \pi - d(y,z)$ and similarly 
$\alpha_y = \pi - d(x,z)$ and 
$\alpha_z = \pi - d(x,y)$.
Thus $\alpha_x + \alpha_y = d(x,y)$ and so on.
One can imagine a triangle drawn in the Euclidean plane with side lengths $d(x,y), d(x,z)$ and $d(y,z)$.
Its inscribed circle would touch the triangle sides $[x,y], [y,z]$ and $[z,x]$ at points $m_z, m_x$ and $m_y$ respectively. From elementary geometry, $[x,y][\alpha_x] = [y,x][\alpha_y]=m_z$ and $[y,z][\alpha_y] = [z,y][\alpha_z]=m_x$ and $[z,x][\alpha_z] = [x,z][\alpha_x]=m_y$, as illustrated in Figure \ref{fig:triangle}.

The points $m_x,m_y,m_z$ are called the \emph{internal points} and $\alpha_x, \alpha_y, \alpha_z$ the \emph{internal distances} corresponding to $x,y,z$ respectively in $\Delta$.  The \emph{insize} of the triangle $\Delta$ is 
the maximum of 
$\max_{\theta \in [0,\alpha_x]} d([x,y][\theta], [x,z][\theta])$, 
$\max_{\theta \in [0,\alpha_y]} d([y,x][\theta], [y,z][\theta])$, and
$\max_{\theta \in [0,\alpha_z]} d([z,x][\theta], [z,y][\theta])$.

\begin{figure}
\begin{center}
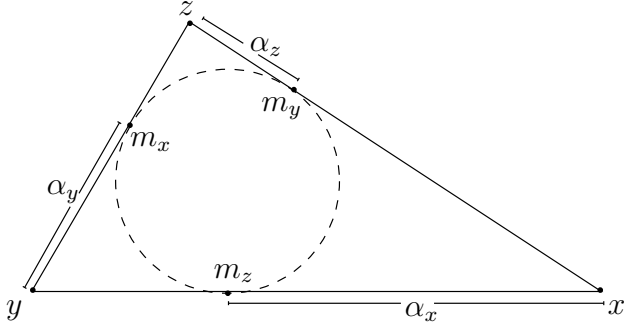
\end{center}
\caption{A geodesic triangle $\Delta(x,y,z)$ with internal points $m_x, m_y$ and $m_z$ and internal distances $\alpha_x, \alpha_y$ and $\alpha_z$ labelled.}\label{fig:triangle}
\end{figure}

\begin{dfn}\label{def:thindelta}
Let $(X,d)$ be a geodesic metric space, and $\delta \ge 0$.
$X$ is $\delta$-hyperbolic (equivalently, the hyperbolicity of $X$ is $\delta$) if the insize of every geodesic triangle is at most $\delta$.
Let $\delta$ be minimum such that the insize of every geodesic triangle is at most $\delta$. We say that $X$ is $\delta$-hyperbolic (equivalently, the hyperbolicity of $X$ is $\delta$).
\end{dfn}

If $G$ is a graph whose associated graphic metric space is $\delta$-hyperbolic then we say $G$ is $\delta$-hyperbolic.
The reader may verify that every tree is $0$-hyperbolic.
Hyperbolicity is sometimes defined in terms of a four-point condition.

\begin{lemma}[4-point condition, see Proposition 1.22 in \cite{BrHa99}]\label{lem:4point}
Let $(X, d)$ be a $\delta$-hyperbolic metric space.
\added{There is a constant $\delta_{4-point} \leq \delta$ such that for}
any 4 points $x,y,z,w \in X$, their ordered set of sums of opposite 
sides, wlog $d(x,y) + d(w,z) \geq d(x,z) + d(y,w) \geq d(x,w) + d(y,z)$, 
satisfy $d(x,y) + d(w,z) - d(x,z) - d(y,w) \leq 2\delta_{4-point}$.
\end{lemma}

The fact that in a $\delta$-hyperbolic metric space 
$\delta_{4-point}$ is always less than or equal to $\delta$ follows 
directly from the proof of Proposition 1.22 on page 411.

\section{$p$-centers and $p$-packings}
\label{sec:ps}
Let $(X,d)$ be a geodesic metric space and $S$ be a compact subset of $X$.
Throughout this paper we rely on two intimately related notions, $p$-centers and $p$-packings.  
\begin{dfn}[$p$-centers]\rm
A set $C\subset X$ $r$-\emph{dominates} $S$ if for every point $s\in S$ there exists a point $c\in C$ with $d(s,c)\leq r$.
The \emph{$p$-radius} of $S$, denoted by $r_p(S)$, is the minimum $r$ such that there exists a set of at most $p$ points  $C_p(S)$ that $r$-dominates $S$.
The points in $C_p(S)$ are called $p$-centers of $S$.
\end{dfn}

\begin{dfn}[$p$-packings]
\rm
A set $D\subseteq S$ is an $r$-\emph{dispersion} in $S$ if each pair of 
points $s, s'  \in D$, $s \neq s'$, $d(s,s') \ge r$.
The \emph{$p$-diameter} of $S$, denoted by $d_p(S)$, is the maximum $r$ such that 
there exists a set of at least $p$ points $D_p(S)$ that is an $r$-dispersion 
in $S$.
The points in $D_p(S)$ are called a {\em $p$-packing.}
\end{dfn}

Consider a set of $p$ points $C$ which $r$-dominate $S$.  
By definition, for any choice of $p+1$ points $D$, each $d \in D$ is within $r$ of some $c \in C$, and by the pigeonhole principle, at least two,  say $a_1$ and $a_2$, are within $r$ of the same $c \in C$.  
Hence, 
\begin{equation*}
d(a_1,a_2) \le d(a_1, c) + d(a_2, c) \le 2r.
\end{equation*}
So, $\min_{i\ne j}d(a_i,a_j) \le 2r$. 
Since this holds for all choices of $C$ and $D$, we have the following observation which first appeared in  \cite{SH77}.
\begin{obs}\label{lowerbound}
$r_p(S) \ge \frac{1}{2} d_{p+1}(S)$.
\end{obs}
It turns out that these two invariants are equal whenever $S$ has a tree-metric.  
Indeed, Shier showed the following.  
\begin{theorem}[Shier \cite{SH77}]\label{cor:shier}
Let $T$ be a tree.
Then $r_p(T) = \frac{1}{2} d_{p+1}(T).$
\end{theorem}

As discussed in Section \ref{sec:hyper}, $\delta$-hyperbolic spaces are treelike,
by which we mean that they possess a metric structure that differs from a tree metric by $\delta$. Therefore, it is logical to attempt to extend Shier's 
result on $p$-center covering and packing to such structures.  
Chepoi and Estellon~\cite{ChEs07} do exactly this by giving an elegant extension of Shier's theorem to $\delta$-hyperbolic spaces.

\begin{theorem}[Chepoi and Estellon\cite{ChEs07}]\label{thm:ches07}
Let $X$ be a $\delta$-hyperbolic metric space and $S$ a finite subset of $X$.
Then $$r_p(S) \leq \tfrac12 d_{p+1}(S) + \delta$$
\end{theorem}
This relationship between $r_p(S)$ and $d_{p+1}(S)$ is a key element in algorithms 
for approximating $p$-centers and $p$-packing.

\subsection{Algorithms for $p$-centers and $p$-packings}
\label{sec:palgs}

The $p$-packing problem, sometimes referred to as the $p$-dispersion problem, has received some attention in the literature. 
For example it is known to be NP-hard~\cite{Erk90}.
Highly relevant to our work is the heuristic that iteratively adds each of the $p$ points by maximizing the points' distance from previously chosen points (see for example \cite{EN91,RRT91}).  
This heuristic is shown to be a $2$-approximation algorithm by Ravi, Rosenkrantz and Tayi~\cite{RRT91}.
For more information, we refer the interested reader to \cite{EUY94} that has an empirical comparison of ten $p$-dispersion heuristics.  

To our knowledge, the previous best algorithm in terms of an additive error not exceeding $\delta$ for the $p$-radius follows from the Chepoi-Estellon bound (Theorem \ref{thm:ches07}).
Indeed, the proof in \cite{ChEs07} leads to a polynomial algorithm to solve $p$-centres in graphs with an additive error of $\delta$ on the $p$-radius.\footnote{The cited result also gives rise to an algorithm for general $\delta$-hyperbolic spaces whose running time depends on the time to compute $F_S(x)$ for $x\in X$ and $S\subseteq X$.
Because our interest is primarily in graphs, we direct the reader to \cite{ChEs07} for details.}
Specifically, in time $O((n^3\log n + n^2m)\log(\diam(G)))$ the authors in \cite{ChEs07} determine a set $U$ of $p$ points such that $U$ $(r_p + \delta)$-dominates $V(G)$.
Their algorithm involves finding diametrical pairs of vertices in subsets of $V(G)$ $O(n\log(\diam(G)))$ times.
Johnson's algorithm \cite{Jo77} finds the diameter in time $O(n^2\log n + nm)$; hence the running time in Chepoi-Estellon \cite{ChEs07} follows.

As pointed out in the introduction, in this work we leverage the fact 
that instead of
finding diametrical pairs, one can just use {\bf locally} diametrical pairs
(introduced in Section \ref{sec:geo}) 
with significant reduction in computational time with only a small penalty 
in the $p$-radius. Our main result is the following.
\begin{theorem}\label{thm:generalp}
Let $G$ be a $\delta$-hyperbolic graph, $p\geq 3$ an integer and $r_p(G)$ the optimal radius of the $p$-center for $V(G)$.
There exists an algorithm to find a set of $p$ points that $(r_p + 3\delta)$-dominates $V(G)$.
Further, the algorithm runs in time $O(n\log n + (m+n)((2p+1)(\lceil 4 + 3\delta + 2\delta \log_2 n \rceil) + (p+1))) = O(p(\delta+1)(m+n)\log n)$.
\end{theorem}

Though the Chepoi-Estellon algorithm~\cite{ChEs07} achieves a better approximation (an additive factor of $\delta$ instead of our $3\delta$), its running time is $O((n^3\log n + n^2m)\log(\diam(G)))$.
We first show below how to improve their running time by a factor of $n$ (Lemma \ref{lem:improve}), but this approach still remains infeasible for large graphs.
When $p\in\{1,2\}$ we can achieve the same Chepoi-Estellon $p$-radius bound but in quasilinear time.

\begin{theorem}\label{thm:smallp}
Let $(X, d)$ be a $\delta$-hyperbolic metric space, $S$ a finite subset of $X$ and $p\in \{1,2\}$.
There exists an algorithm to determine a set of $p$ points that $(r_p + \delta)$-dominate $S$.
Further, the algorithm runs in 
time $O((2\delta+1)t_X)$, where $t_X$ is the time required to find the set of points at maximum distance from a given point in $X$.
In particular in a $\delta$-hyperbolic graph the running time is $O((2\delta + 1)(m+n))$.
\end{theorem}
For $p=1$, the previous best algorithm we know of is due to Chepoi et al. 
\cite{ChDrEsHaVa08}: the approximation error is $\leq 5\delta$, and the 
computation requires just two breadth-first searches.
In contrast, we require $2\delta+1$ breadth-first searches to achieve the 
smaller additive factor of $\delta$.

The remainder of this section is organized as follows.  
We start by showing how to improve the time complexity of the Chepoi-Estellon 
algorithm by only approximately finding diametrical 
pairs of vertices, that is via finding locally diametrical pairs.  
In the proofs of our main results, we will repeatedly apply this idea, showing that it is sufficient to solve the easier and computationally more efficient approximate version of this expensive sub-problem. 
We then move on to proofs of Theorems \ref{thm:smallp} and \ref{thm:generalp}  
in Sections \ref{sec:smallp} and \ref{sec:generalp}, respectively.

Recall from Section \ref{sec:geo} that a pair of vertices $\{u,v\}$ is 
{\em locally diametrical} if there is no vertex $w$ such that $d(u,w)>d(u,v)$ 
or $d(v,w)>d(v,u)$. 
Clearly a diametrical pair is locally diametrical but the converse is 
not true in general (e.g., a cycle with handles). 
It turns out to be sufficient to find locally diametrical pairs in 
the main lemma of \cite{ChEs07}. 
Indeed, the following lemma is simply Lemma 1 from \cite{ChEs07}, but with the requirement that $u$ and $v$ be diametrical replaced with the weaker property of being locally diametrical.
\begin{lemma}\label{lem:locallyopposite}
Let $X$ be a $\delta$-hyperbolic metric space and $S\subseteq X$ be a 
compact set and $r\in \mathbb{R}$.
Suppose that $u$ and $v$ are locally diametrical in $S$ and let 
$[u,v]$ be a geodesic.
Let $c = [u,v][r]$.
Then $B_{2r}(u) \cap S \subseteq B_{r+\delta}(c) \cap S$.
\end{lemma}
The proof of Lemma 1 in \cite{ChEs07} works 
essentially unchanged to prove Lemma \ref{lem:locallyopposite} by replacing diametrical 
pairs with locally 
diametrical pairs. 
Since we will use a refined version of the same argument 
that is needed for Lemma \ref{lem:locallyopposite} in the proof of 
Theorem \ref{thm:generalp}, 
we skip the proof of Lemma \ref{lem:locallyopposite}.
We prove below (Lemma \ref{lem:locallyapart}) that we can find a locally 
diametrical pair with at most $2\delta+1$ breadth-first searches.  
Hence, we achieve the following significant reduction in the 
run time of the Chepoi-Estellon algorithm.
\begin{lemma}\label{lem:improve}
Let $G$ be a $\delta$-hyperbolic graph and $p$ an integer.
There exists an algorithm to find a set of $p$ points that 
$(r_p + \delta)$-dominates $V(G)$ that runs in time 
$O(n^2 \log(\diam(G)) (2\delta+1))$.
\end{lemma}

It remains to show how to efficiently determine locally diametrical pairs.
\begin{lemma}\label{lem:locallyapart}
Given a $\delta$-hyperbolic graph $G$ and $S\subseteq V(G)$.
There is an algorithm that finds a locally diametrical pair of vertices
by performing at most $2\delta + 1$ breadth-first searches; that is, the running time is $O((2\delta+1)(m+n))$.
\end{lemma}

\begin{proof}
Choose a vertex $u\in S$ arbitrarily and find a vertex $v_1\in F_S(u)$ by BFS. 
Then, find $v_2\in F_S(v_1)$. 
Next, find a vertex $v_3\in F_S(v_2)$. 
If $d(v_2,v_3) = d(v_1,v_2)$, then let $v = v_1$ and $w = v_2$ and we have
found a locally diametrical pair. 
Otherwise $d(v_2,v_3) > d(v_1,v_2)$ and continue the process until 
$v_k,v_{k+1}$ are found such that $d(v_k,v_{k+1}) = d(v_k, F_S(v_k))$ and 
$d(v_k,v_{k+1}) = d(v_{k+1}, F_S(v_{k+1}))$.
This must happen for at most $k \leq diam(S)$. But by Proposition 3
in \cite{ChDrEsHaVa08} 
$d(v_1,v_2) \geq diam(S)-2\delta_{4-point} \geq diam(S)-2\delta$ 
so $k$ cannot
exceed $2\delta$.  This means no more than $(2\delta+1)$ BFS steps or
no more than $O(2\delta+1)(m+n)$ steps are needed 
for finding a locally diametrical 
pair starting from $u\in S$.
Then algorithm returns the locally diametrical pair $(v_k, v_{k+1})$. 
\end{proof}

\section{Approximating $p$-centers}
\label{sec:smallp}

In general, in searching for $p$-centers, first we approximately solve the dual problem, that is, we find $D$, a $(p+1)$-packing, with $|D| \geq p+1$ 
such that
$$ \{\max~r~|~d(s,s') \ge r,~\forall s \ne s' \in D\} \le d_{p+1}(V).$$
This together with Observation \ref{lowerbound} yields 
\begin{equation}
\label{eqn:yields}
 \frac{1}{2} \{\max~r~|~d(s,s') \ge r,~\forall s \ne s' \in D\} \le r_p(V).
 \end{equation}
Given these $(p+1)$-points we find a set of $p$-points $C$ such that setting $\lambda =   \tfrac{1}{2} \{\max~r~|~d(s,s') \ge r,~\forall s \ne s' \in D\}$,
\begin{enumerate}
\item $C$ $\lambda$-dominates the points in $D$, and
\item for each $a \in D$ there exists some $a' \in D$ and $c \in C$ such that $c$ is on a geodesic between $a$ and $a'$.  
\end{enumerate}
We prove later that these two properties together with $\delta$-hyperbolicity allow us to show that for a carefully-selected set $D$, 
the $p$ points in $C$ $(\lambda + 3\delta)$-dominate $V$, that is,
\begin{equation}
\label{eqn:bound}
\{\min~r~|~for~each~x\in V,~ \exists c\in C~with~d(x,c) \le r\} \le \lambda + 3\delta.
\end{equation}
Substituting the value of $\lambda$ in (\ref{eqn:bound}) and applying (\ref{eqn:yields}) yields,
\begin{eqnarray*}
\{\min~r~|~for~each~x\in V,~ \exists c\in C~with~d(x,c) \le r\}
&\le& \frac{1}{2} \{\max~r~|~d(s,s')\ge r,~\forall s \ne s' \in D\} + 3\delta \\
&\le&  r_p(V) + 3\delta.
\end{eqnarray*}
It follows that $C$ $(r_p(V) + 3\delta)$-dominates $V$ as desired.

We now apply this approach to find a $1$-center of a graph.  

\begin{theorem}\label{thm:1-center}
Let $G$ be a $\delta$-hyperbolic graph.
There exists an algorithm to find a point $c$ that $(r_1 + \delta)$-dominates $V(G)$.
The algorithm requires time $O((2\delta + 1)(m+n))$.
\end{theorem}

\begin{proof}
Let $x,y$ be a locally diametrical pair of vertices and let $[x,y]$ be a geodesic segment.
As described above, set $\lambda = \tfrac{d(x,y)}{2}$ and choose $c = [x,y][\lambda]$.
Clearly, $C = \{c\}$ satisfies Properties 1 and 2 above.  
We now show that $C = \{c\}$ $(\lambda + \delta)$-dominates $V$.  

Let $z$ be any point in $V$ and consider the geodesic triangle $\Delta(x,y,z)$ as depicted and labeled in Figure \ref{fig:triangle2}.
Without loss of generality, assume that  $d(y,z) \le d(x,z)$.
Since $(x,y)$ is locally diametrical, then 
\begin{eqnarray*}
d(y,z) \leq d(x,z) \leq d(x,y)
\end{eqnarray*}
which implies that 
\begin{eqnarray*}
\alpha_z \leq \alpha_y \leq \alpha_x.
\end{eqnarray*}
(This means that in the figure $c$ lies to the right of $m_z$, as
shown.) Then
\begin{eqnarray*}
d(z,c) \leq \alpha_z + \delta + d(c,m_z) \leq \alpha_z + \delta + \lambda - \alpha_y \leq \delta + \lambda.
\end{eqnarray*}

\begin{figure}
\begin{center}
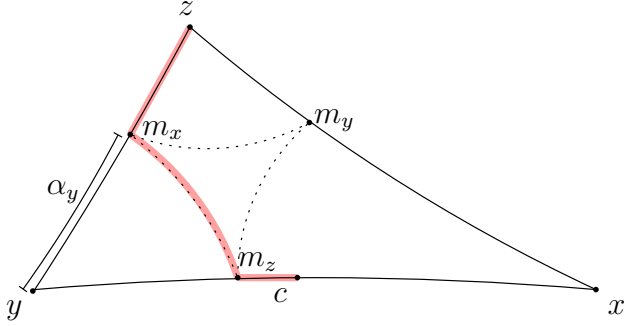
\end{center}
\caption{A geodesic triangle $\Delta(x,y,z)$ with internal points $m_x, m_y, m_z$ and $c$ labelled as in the proof of Theorem \ref{thm:1-center}. 
Dashed lines indicate a distance $\leq \delta$ and the red line indicates the upper estimate for $d(z,c)$.
}
\label{fig:triangle2}
\end{figure}

As the claim holds for any $z$, $c$ $(\lambda+\delta)$-dominates $V(G)$, and therefore, 
since $\lambda= \tfrac{1}{2}d(x,y) \leq \tfrac{1}{2} d_2(V) \leq r_1(V)$, the latter inequality
by Observation \ref{lowerbound}, and thus
$c$ $(r_1 + \delta)$-dominates $V(G)$, as desired. 
To complete the proof, we note that by Lemma \ref{lem:locallyapart}, $x$ and $y$ can be found in time $O((2\delta+1)(m+n))$.
\end{proof}
We note that in the course of the above prove we demonstrated the following fact
that we shall reuse.
\begin{obs}\label{obs:dominate}
Let $z$ be any vertex in $V(G)$, $(x,y)$ a locally diametrical pair of vertices, $c\in A(G)$ the mid-point of $[x,y]$ and $\lambda = \frac{d(x,y)}{2}$. 
Then $d(z,c) \le \lambda + \delta$.
\end{obs}

In extending these proof techniques to the general case for $p > 1$, we run into the following two difficulties, each costing us an additional $\delta$ in our approximation error.     
First, Property 2 only guarantees that $p$ of the $\binom{p+1}{2}$ pairs of points in $D$ have a geodesics connecting them containing some point $c_i \in C$. 
This will force us use two geodesic triangles to bound the distance from some points in $V$ to their closest center in $C$.  
Second, in achieving the quasilinear runtime, we are only able to find a $(\lambda+2\delta)$-approximation for the $(p+1)$-packing problem. 
We omit further details until Section \ref{sec:generalp}.

To finish off this section, we prove that when $p = 2$ we can find a 2-center solution which $(r_2 + \delta)$-dominates $G$.  
Like Theorem \ref{thm:1-center}, this is stronger than our general result (Theorem \ref{thm:generalp}) and the proof does
not use the machinery outlined at the beginning of Section \ref{sec:smallp} that relies on Properties 1 and 2. Theorems \ref{thm:1-center} and \ref{thm:p2} 
may be special cases of a general and stronger result
than our main result, so we include it.

\begin{theorem}\label{thm:p2}
Let $G$ be a $\delta$-hyperbolic graph.
There exists an algorithm to find points $c_1, c_2$ that $(r_2 + \delta)$-dominate $V(G)$.
The algorithm requires time $O((2\delta+1)(m+n))$.
\end{theorem}

\begin{proof}
Let $x,y$ be a locally diametrical pair of vertices and let $[x,y]$ be a geodesic segment.
Choose $z$ so that $\min\{d(z,x), d(z,y)\}$ is maximized (requires two BFS).
We let our $3$-packing be $D = \{x,y,z\}$.
Assume without loss of generality that $d(x,y)\geq d(x,z) \geq d(y,z)$, and so, $\lambda =  \tfrac{1}{2} \{\max~r~|~d(s,s') > r,~\forall s \ne s' \in D\} = \tfrac12 d(y,z)$.

We choose $c_1 = [x,y][\lambda]$ and $c_2 = [y,x][\lambda]$.
We claim that $C = \{c_1,c_2\}$ satisfy Equation \ref{eqn:bound}, with $t = 1$, and so, $C$ $(r_2 + \delta)$-dominates $G$.  

To prove the claim, let $\Delta_1 = \Delta(x,y,z)$ be a geodesic triangle.
Let $w$ be any point of $G$ and let $\Delta_2 = \Delta(x,y,w)$ be a geodesic triangle so that $\Delta_1$ and $\Delta_2$ share the geodesic $[x,y]$.
We will show that $\min\{d(w,c_1), d(w,c_2)\} \leq \lambda + \delta$.
Take $\alpha_x, \alpha_y, \alpha_w$ and $m_x, m_y, m_w$ to denote the internal distances and points in $\Delta_2$.
Without loss of generality assume $d(w,x) \leq d(w,y)$ which implies that
$d(w,x) \leq d(y,z) = 2 \lambda$  and $\alpha_x \leq \alpha_y$.  
We distinguish two cases, as illustrated in Figure \ref{fig:p2}.
\begin{figure}[h]
\begin{center}
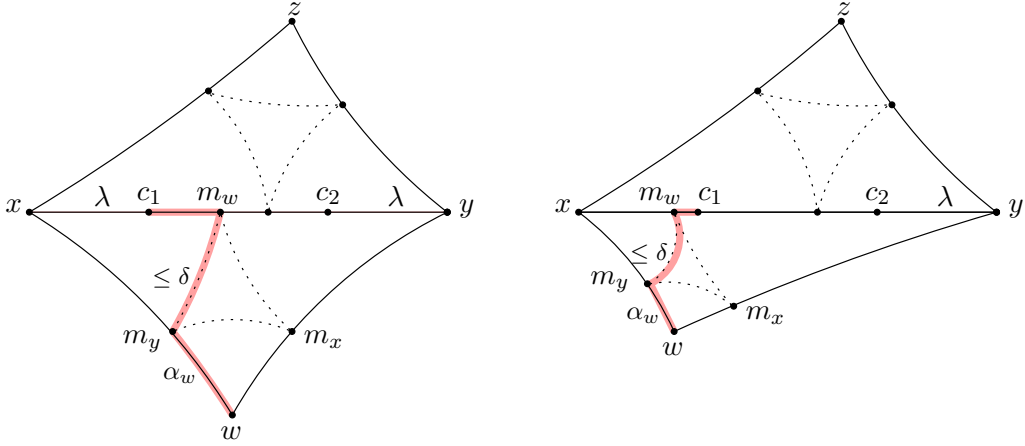
\end{center}
\caption{Figure for Cases 1 and 2 in the proof of Theorem \ref{thm:p2}. The red lines indicate the upper estimate for $d(w,c_1)$. Dashed lines indicate a distance $\leq \delta$.}\label{fig:p2}
\end{figure}

\vspace{.5em}\noindent{\bf Case 1}: $\lambda < \alpha_x < d(x,y)-\lambda$ \\
\noindent From the choice of $z$, it follows that either $d(w,x) \leq d(y,z) = 2\lambda$ or $d(w,y) \leq 2\lambda$.
Assume without loss of generality that $d(w,x) = d(w,m_y) + d(m_w,x) \leq 2\lambda$.
Therefore, $d(w,c_1) \leq d(w, m_y) + d(m_y,m_w) + d(m_w, c_1) \leq d(w, m_y) + \delta + d(m_w,x) - \lambda \leq \lambda + \delta$.

\vspace{.5em}\noindent{\bf Case 2}: $\alpha_x \leq \lambda$ \\
\noindent In this case $m_w$ lies between $x$ and $c_1$ on the geodesic segment $[x,y]$.
By the local maximality of $x$ and $y$, we have $d(y,w) = \alpha_y + \alpha_w \leq \alpha_y + \alpha_x = d(x,y)$ and so $d(w,m_y) = \alpha_w \leq \alpha_x = d(x,m_w)$.
Then $d(w, c_1) \leq d(w,m_y) + d(m_y, m_w) + d(m_w, c_1) \leq d(x,c_1) + \delta = \lambda + \delta$.

To complete the proof, we need only show that $c_1$, $c_2$ can be found in $O((2\delta+1)(m+n))$ time.  
By Lemma \ref{lem:locallyapart}, $x$ and $y$ can be found in time $O((2\delta+1)(m+n))$ and the vertex $z$ can be found by doing a breadth-first search rooted at $x$ and one rooted at $y$.
Given $D = \{x,y,z\}$, the vertices $c_1$ and $c_2$ can then be found by storing the last breadth-first search used in finding $x$ and $y$ and $\lambda = \tfrac 12 \min \{d(x,x),d(y,z)\}.$
The runtime now follows.
\end{proof}

\subsection{The general algorithm}\label{sec:generalp}

Our algorithm and proof follow the same three basic steps, though each step is more involved.  
As a reminder these three steps are 
1) approximately solving the dual problem, or finding a $(p+1)$-packing, 
2) deriving $p$-points from this dual solution that satisfy Properties 1 and 2, and 
3) bounding the approximation guarantee by showing Equation \ref{eqn:bound}.

It turns out the difficult part of these three steps is Step 1.  
For this step, we need to extend the notion of a `locally diametrical pair' to a `locally diametrical set' in such a way that i) it provides us with both the tools we need to satisfy Properties 1 and 2 and ii) it can be determined efficiently.  
We find a set of $(p+1)$ vertices $D = \{v_0, v_1, ..., v_p\}$ with
$$\lambda(D) :=   \tfrac{1}{2} \{\max~r~|~d(s,s') \ge r,~\forall v_i \ne v_j \in D\}$$ such that the following three properties hold 
\begin{itemize}
\item [(a)] (Vertex relabeling) $d(v_0,v_i) = 2\lambda(D)$ for some $v_i \in D$, 
\item [(b)] (Extending locally diametrical pairs to \emph{locally diametrical sets}) For each $v_i \in D$ with $d(v_i,v_j) = 2\lambda(D)$ for some $v_j$, there exists no $w\in V(G)$ with $d(w, v_k) > 2\lambda(D), \forall  v_k \in D \setminus \{v_i\}$, and
\item [(c)] ($\delta$-hyperbolic version of locally diametrical sets) for each $i \geq 1$, there exists no vertex $v\in V(G)$ with $d(v_0,v) > d(v_0, v_i) + 2\delta$ and $d(v_i,v) \leq 2\lambda(D)$ and $d(v,v_j) > 2\lambda(D)$ for each $j\neq i$.
\end{itemize}
These three requirements provide us with what is needed to determine a set of $(p+1)$ vertices satisfying  Properties 1 and 2.  
Specifically, we prove
\begin{lemma}\label{lem:optimal-optimized}
Let $G$ be a $\delta$-hyperbolic graph and $\Lambda_n = \lceil 4 + 3\delta + 2\delta \log_2 n \rceil$.  
There exists an algorithm to find a set $D$ of $p+1$ vertices satisfying (a), (b) and (c).
The algorithm runs in time $O(n\log n + (m+n)((2p+1)\Lambda_n + (p+1)))$.
\end{lemma}

Given a set of $p+1$ vertices satisfying Properties (a), (b) and (c) it is straightforward to find $C = \{c_1, ..., c_p\}$ satisfying Properties 1 and 2.
For each $1\leq i\leq p$,  let $c_i$ be the vertex at distance $\lambda$ from $v_i$ on the shortest path from $v_i$ to $v_0$, i.e. $c_i = [v_i, v_0][\lambda]$.

\begin{lemma}\label{lem:finalsolution}
Let $G$ be a $\delta$-hyperbolic graph.
Suppose that $D = \{v_0, v_1, ..., v_p\}$ satisfy (a), (b) and (c).  
Then the set of $p$ points $C = \{c_i ~|~ c_i = [v_i, v_0][\lambda]\}$
$(\lambda + 3\delta)$-dominate $G$. 
\end{lemma}

As described above (beginning of Section 4), 
such $C$ $(r_p(V) + 3\delta)$-dominates $V$ as desired.  
So, given the Lemmas \ref{lem:optimal-optimized} and \ref{lem:finalsolution}, the proof of Theorem \ref{thm:generalp} follows once establishing the runtime, which we do now.  
First, determining the set $D$ takes $O(n\log n + (m+n)((2p+1)\Lambda_n + (p+1)))$.  
Given $D$, the set of vertices $\{c_i, 1\leq i \leq p\}$ can clearly be constructed by performing a breadth-first search rooted at $v_0$.
Theorem \ref{thm:generalp} now follows.  

In the next two sections we establish Lemmas  \ref{lem:optimal-optimized} and \ref{lem:finalsolution}.
Lemma \ref{lem:optimal-optimized} is the more interesting of the two proofs, and takes us deeper into the analysis of locally diametrical sets.
The proof of Lemma \ref{lem:finalsolution} is a sophistication of the ideas in Theorems \ref{thm:1-center} and \ref{thm:p2}. We begin with that lemma.

\subsection{Proof of Lemma \ref{lem:finalsolution}}
We show that every vertex of $G$ is at distance at most $\lambda + 3\delta$ from some centre $c_i$.
Let $w \in V(G)$ and suppose that $w$ is at distance greater than $\lambda + 3\delta$ from each centre.
Property (b) implies $d(w,v_i)\leq 2\lambda$ for some $i$.
We prove below the following claim.  

\begin{claim}\label{clm:opt}
$d(w,v_j) > 2\lambda$ for each $j\neq i$.
\end{claim}

Using the claim, we can prove Lemma \ref{lem:finalsolution}. 
Consider the geodesic triangle $\Delta(v_i,v_0,w)$, and recall that $c_i$ belongs to the geodesic $[v_i,v_0]$.
There are two cases to handle.

First, suppose that $d(v_i,m_w) \geq \lambda$.
Then a $w$-$c_i$-path can be constructed by concatenating the geodesics from $[w,m_{v_0}], [m_{v_0},m_{w}]$ and  $[m_w,c_i]$, and so, since $d(w,v_i)\leq 2\lambda$
\begin{eqnarray*}
d(w,c_i) &\leq& d(w, m_{v_0}) + d(m_{v_0},m_{w})+ d(m_w, c_i) \\
&\leq& d(w, m_{v_0}) + \delta +  d(m_w, v_i) - \lambda\\
&\leq&  \lambda + \delta,
\end{eqnarray*} 
a contradiction.

Otherwise, if $ d(v_i,m_w) < \lambda$, then
\begin{eqnarray*} 
\lambda + 3\delta < d(w,c_i) &\leq& d(w,m_{v_0}) + d(m_{v_0},m_w)  + d(m_w, c_i) \\
&\leq& d(w,m_{v_0}) + \delta + d(m_w, c_i).
\end{eqnarray*} 
Since $d(v_i, c_i) = d(v_i, m_w) + d(m_w, c_i) = \lambda$, we deduce that $d(w,m_{v_0}) > d(v_i,m_w) + 2\delta$.
It follows that $d(v_0, w) > d(v_0,v_i) + 2\delta$, which along with Claim \ref{clm:opt}, contradicts Property (c).

It follows that $w$ is within $\lambda + 3\delta$ from at least one centre.
We need only prove the claim.

\begin{proof}[Proof of Claim \ref{clm:opt}]
Suppose that $w$ is at distance at most $2\lambda$ from both $v_i$ and $v_j$.
Let $c_i'$ and $c_j'$ be the vertices at distance $\lambda$ from $i$ and $j$ respectively on the geodesic $[v_i, v_j]$.
We will show that at least one of $d(c_i,c_i')$ and $d(c_j,c_j')$ is at most $\delta$.
Consider the geodesic triangle $\Delta(v_i,v_j,v_0) = [v_i, v_j]\cup[v_i, v_0]\cup [v_j, v_0]$ and let $m_{v_i}, m_{v_j}, m_{v_0}$ be as described above.
Assume for contradiction that both $d(c_i,c_i')$ and $d(c_j,c_j')$  are greater than $\delta$.
It follows that  $d(v_i, m_{v_j}) < \lambda$ and $d(v_j, m_{v_i}) < \lambda$. 
But then $d(v_i,v_j) = d(v_i, m_{v_0}) + d(m_{v_0}, v_j) = d(v_i, m_{v_j}) + d(v_j, m_{v_i}) < 2\lambda$, a contradiction.
Assume then, without loss of generality, that $d(c_i,c_i') \leq \delta$.

Now consider the geodesic triangle $\Delta(v_i, v_j, w)$ and let $m_w$ be defined as usual.
First, suppose that $d(v_i,m_w) \geq d(v_i, c_i')$.
Then 
$$d(w, c_i') \leq d(w, m_{v_j}) + \delta + d(v_i, m_w) - d(v_i, c_i') \leq d(w, v_i) + \delta - \lambda \leq \lambda + \delta.$$
Now, suppose that $d(v_i,m_w) < d(v_i, c_i')$.
Then 
$$d(w, c_i') \leq d(w, m_{v_i}) + \delta + d(v_j, m_w) - d(v_j, c_i') \leq d(w, v_j) + \delta - \lambda \leq \lambda + \delta.$$
In either case, $d(w,c_i') \leq \lambda + \delta$, and so $d(w,c_i) \leq \lambda + 2\delta$, a contradiction.
\end{proof}

\subsection{Proof of Lemma \ref{lem:optimal-optimized}}

A proof sketch is as follows.  
We first show that we can a find $(p+1)$-packing that is within $O(\delta \log_2 n)$ of an optimal solution.  
To do so, we find a tree $T$ which approximately preserves distances on our input graph $G$.  
It turns out that exactly solving the $(p+1)$-packings on trees can be done efficiently, though in contrast to before, we solve the $p$-centres first and use this to construct a dual solution in $G$.  
The fact that $T$ is a good approximating tree allows us to bound how close our $(p+1)$-packing is to an optimal solution and in turn helps us achieve the quasilinear running time.  
Finally, given this initial $(p+1)$-packing, we iteratively improve the solution whenever possible until we achieve Properties (a), (b), (c).  
Clearly, (a) can hold for all solutions after relabelling, so the only difficulty is in insuring both (b) and (c) hold.  

We will use the following theorem, which we will deduce from known results at the end of this section, to find our initial $(p+1)$-packing. 
Let $\Lambda_n = \lceil 4 + 3\delta + 2\delta \log_2 n \rceil$.

\begin{theorem}\label{cor:dispersion}
There exists an algorithm to find a set $\mathcal P$ of $p+1$ vertices satisfying $d(u,v)\geq \kappa, \forall u\neq v \in \mathcal P$, for some $\kappa$ with $d_{p+1}(G) - \kappa \leq \Lambda_n$. 
The algorithm runs in time $O(n \log n)$.
\end{theorem}

Given the set $\mathcal P$ of  $(p+1)$-points from Theorem \ref{cor:dispersion}, we now describe an efficient iterative algorithm which finds $(p+1)$-points satisfying Properties (a), (b) and (c).  
Our argument bounds the number of iterations using the following potential function.

\begin{dfn}
Let $G$ be a graph and let $\mathcal P \subseteq V$ be a set of $p$ vertices and suppose that $\kappa$ is the largest value such that $d(u,v) \ge \kappa$ for all $u\ne v \in \mathcal P$.  Let $\eta(\mathcal P)$ denote the number of vertices in $\mathcal P$ which are exactly at distance $\kappa$ from at least one other vertex in $\mathcal P$.
We define the potential of $\mathcal P$ as $\phi(P) := p(\kappa + 1) - \eta(\mathcal P)$.
\end{dfn}

\begin{algorithm}[ht!]
\LinesNotNumbered
\DontPrintSemicolon
\KwIn{Graph $G = (V,E)$ and integer $p$.}
\KwOut{A set $\mathcal P$ of $p+1$ vertices satisfying $d(u,v)\geq \kappa, \forall u\neq v \in \mathcal P$, for some $\kappa$ with $d_{p+1}(G) - \kappa \leq \Lambda_n$.}
{\bf Let} $T = (V,F)$ be the tree determined by Theorem \ref{thm:gromovtree}. \\
{\bf Let} $\lambda$ be the $p$-radius of the $p$-centers of $T$ determined by Theorem \ref{thm:linearpcentre}. \\
{\bf Let} $\mathcal P$ be a set of maximum size s.t. $d(u,v) \ge 2\lambda$ for each $u\neq v \in \mathcal P$ (Theorem \ref{thm:dispersion}). \\
{\bf Let} $\mathcal P'$ be a set of $p+1$ unique vertices chosen arbitrarily from $\mathcal P$.\\
\Return {$\mathcal P = \mathcal P'$}
\caption{Finding an initial set $\mathcal P$ of vertices\label{alg:initial}}
\end{algorithm}

\begin{algorithm}[ht!]
\SetKw{KwFrom}{from}
\LinesNotNumbered
\DontPrintSemicolon
\SetAlgorithmName{Subroutine}{subroutine}{}
\KwIn{A set $\mathcal P$ satisfying Property (a) for $\lambda(\mathcal P)$.}
\KwOut{A set $\mathcal P'$ satisfying Property (a) and (b) for $\lambda(\mathcal P') \ge \lambda(\mathcal P)$.}
We say that a vertex $u\in \mathcal P$ is \emph{improvable} to $w\notin \mathcal P$ if there exists $v\in \mathcal P$ with $d(u,v) = \added{2}\lambda(\mathcal P)$ and  $d(w,x) > \added{2}\lambda(\mathcal P), \forall x \in \mathcal P\setminus \{u\}$.

\Repeat{no vertex is found to be improvable.}{
\For{$i$ \KwFrom $0$ \KwTo $p$}{
\uIf {$v_i$ is improvable to some $v$} 
{
replace $v_i$ in $\mathcal P$ with the improved vertex ($\mathcal P = (\mathcal P \setminus{v_i}) \cup \{v\}$).\;
}
\uElse{
	do nothing.
}
}
}
\Return{$\mathcal P' = \mathcal P$.}
 \caption{Satisfying Properties (a) and (b).\label{alg:potential}}
\end{algorithm}

\begin{algorithm}[ht!]
\LinesNotNumbered
\DontPrintSemicolon
\SetAlgorithmName{Subroutine}{subroutine}{}
\SetKw{KwFrom}{from}
\KwIn{A set $\mathcal P$ satisfying Property (a) and (b) for $\lambda(\mathcal P)$.}
For this step we label the vertices of $\mathcal P$ in a specific way.
Let $v_0$ and $v_p$ be vertices in $\mathcal P$ with $d(v_0,v_p) = \added{2}\lambda(\mathcal P)$.
Then label the remaining vertices of $\mathcal P$ as $\{v_0, v_1, \dots, v_p\}$ so that $d(v_0, v_i) \geq d(v_0, v_j)$ for each $i>j$.

In this context, we say that a vertex $v_i$ ($1\leq i \leq p$) is improvable to $v_i' \notin \mathcal P$ if 
$d(v_0,v_i') > d(v_0, v_i)$ and $d(v_i,v_i') \leq \added{2}\lambda(\mathcal P)$ and $d(v_i',v_j) >  \added{2}\lambda(\mathcal P)$ for each $j\neq i$.

\For{$i$ \KwFrom $1$ \KwTo $p$}{
\uIf {$v_i$ is improvable}
{
replace $v_i$ in $\mathcal P$ with the vertex $v_i'$ furthest from $v_0$ that satisfies $d(v_i,v_i') \leq \added{2}\lambda(\mathcal P)$ and $d(v_i',v_j) > \added{2}\lambda(\mathcal P)$ for each $j\neq i$.\;
}
}

\Return{$\mathcal P$.}
\caption{Satisfying Properties (a), (b) and (c).\label{alg:vertex}}
\end{algorithm}

\begin{algorithm}[ht!]
\LinesNotNumbered
\DontPrintSemicolon
\KwIn{A set $\mathcal P$ satisfying Property (a) for $\lambda(\mathcal P)$.}
\KwOut{A set $\mathcal P'$ satisfying Property (a), (b) and (c) for $\lambda(\mathcal P') \ge \lambda(\mathcal P)$.}
{\bf Let} $\mathcal P$ be the returned set of Subroutine \ref{alg:potential} with input $\mathcal P$. \\
\Repeat{$\phi(\mathcal P') = \phi(\mathcal P)$}{
	{\bf Let} $\mathcal P'$ be the returned set of Subroutine \ref{alg:vertex} with input $\mathcal P$.  \\
	{\bf Let} $\mathcal P$ be the returned set of Subroutine \ref{alg:potential} with input $\mathcal P'$. 
}
\Return {$\mathcal P' = \mathcal P$}
\caption{Finding an optimal and optimized set of vertices\label{alg:optopt}}
\end{algorithm}

Algorithm \ref{alg:optopt} together with Subroutines \ref{alg:potential} and \ref{alg:vertex} describe the algorithm.
We first prove that if Algorithm \ref{alg:optopt} terminates then it is correct, that is, $\mathcal P'$ satisfies (a), (b) and (c).
The algorithm terminates if the potential $\phi(\mathcal P)$ has not increased \added{after successive executions of Subroutines \ref{alg:vertex} and \ref{alg:potential} }. 
As Subroutine \ref{alg:potential} executes last, the returned $\mathcal P$ satisfies (a) and (b) as satisfying (b) is the stopping condition and, as mentioned above, (a) always holds after a relabelling. 
\added{Since $\phi(\mathcal P)$ is unchanged by Subroutine \ref{alg:potential}, $\mathcal P$ is unchanged as well.}
For the purpose of analysis, we will adopt the following notation.
Let $\{v_0,\dots, v_p\}$ be the labelling specified in the description of Subroutine \ref{alg:vertex}.
Then, for each $v_i (i\geq 1)$, if it was improved, let $v_i'$ be the vertex $v_i$ was replaced by.
Otherwise write $v_i' = v_i$.
So, now consider $\mathcal P = \{v_0, v_1', \dots, v_p'\}$, that is output by Subroutine \ref{alg:vertex}.
For contradiction, suppose that $\mathcal P$ does not satisfy Property (c).
Then, there exist an index $i$ and a vertex $v_i''$ with $d(v_0, v_i'') > d(v_0, v_i') + 2\delta$ and $d(v_i', v_i'') \leq \added{2}\lambda(\mathcal P)$ and $d(v_i'', v_j') \geq \added{2}\lambda(\mathcal P), \forall j\neq i$.
Further, by the choice of $v_i'$, there must exist some index $j>i$ with $d(v_i'', v_j) < \added{2}\lambda(\mathcal P)$.
We will apply Lemma \ref{lem:4point} to reach a contradiction, using the vertices $v_0,v_i',v_i'',v_j$, as illustrated in Figure \ref{fig:4point}.
By the choice of labelling, we have
$\added{2}\lambda(\mathcal P) \leq d(v_0,v_j) \leq d(v_0,v_i) \leq d(v_0, v_i') < d(v_0, v_i'') - 2\delta$.
There are three distance sums to consider.
We claim that $d(v_0, v_i'') + d(v_i',v_j) > \max\{d(v_0, v_i') + d(v_i'',v_j) , d(v_0,v_j) + d(v_i',v_i'') \} + 2\delta$.
This is clear because both $d(v_i'',v_j)\leq \added{2}\lambda(\mathcal P)$ and $d(v_i',v_i'') \leq \added{2}\lambda(\mathcal P)$ while $d(v_i', v_j) \geq \added{2}\lambda(\mathcal P)$.
By Lemma \ref{lem:4point}, this contradicts the $\delta$-hyperbolicity of $G$.
It follows that $\mathcal P$ also satisfies (c).

\begin{figure}[h]
\begin{center}
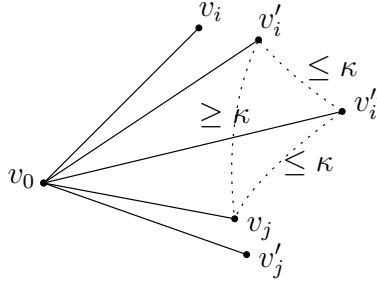
\end{center}
\caption{Proof that $\mathcal P$ satisfies Property (c)}\label{fig:4point}
\end{figure}

It remains to prove that the algorithm terminates and to bound the runtime.  
To see that Algorithm \ref{alg:optopt} terminates, we first note that whenever a vertex in $\mathcal P$ is improved in Subroutine \ref{alg:potential}, the distance to its closest neighbour strictly increases.
Therefore, after at most $p+1$ rounds of the repeat until loop Subroutine \ref{alg:potential}, $\lambda(\mathcal P)$ strictly increases.
Further, each round (except the last one) in which the potential doesn't change is proceeded by an iteration of Subroutine \ref{alg:vertex}.
By Theorem \ref{cor:dispersion},  for the initial $(p+1)$-points $\mathcal P^\star$, $d_{p+1}(G) - \added{2}\lambda(\mathcal P^\star) \leq \Lambda_n$.
Hence, $\phi(\mathcal P^\star) \geq (p+1)(\added{2\lambda(\mathcal P^\star)} + 1) - (p+1) = (p+1)\added{2\lambda(\mathcal P^\star)} \geq (p+1)(\added{d_{p+1}}(G) - \Lambda_n)$.
Further, any set of $p+1$ vertices has dispersion at most $\added{d_{p+1}}(G)$ and therefore has potential at most $(p+1)(\added{d_{p+1}}(G) + 1)$.
We conclude that the repeat until loop of  Algorithm \ref{alg:optopt} can be executed at most $(p+1)\Lambda_n$ rounds in total.
%
%


We now examine the complexity of the algorithm.
To obtain the initial set $\mathcal P$ as in Theorem \ref{cor:dispersion} takes time $O(n \log n)$.
Given a set $\mathcal P$, we can determine and record the set of distances $\{d(v,v_i): v\in V(G), 0\leq i \leq p \}$ by performing a breadth-first search rooted at each vertex $v_i\in \mathcal P$.
From these distances, it can easily be checked in linear ($O(n)$) time whether a vertex is improvable.
To complete the first round the first time we perform Subroutine \ref{alg:potential}, we must perform $p+1$ breadth-first searches.
Each time a vertex is improved (in either Subroutine \ref{alg:potential} or Subroutine \ref{alg:vertex}), we need an additional one.
From the discussion above it follows that at most $p+1 + ((p+1) + p)\Lambda_n$ breadth-first searches need be done.
The algorithm therefore runs in time $O(n\log n + (m+n)((2p+1)\Lambda_n + (p+1)))$.

We now deduce Theorem \ref{cor:dispersion} and its corresponding Algorithm \ref{alg:initial}. 
In finding our initial $(p+1)$-packing, we use the following definitions and results.
For a graph $G$ and constant $k$, we say that a tree $T$ with vertex set $V(G)$ is a $k$-approximating tree if $|d_G(u,v) - d_T(u,v)| \leq k$ for every pair of vertices $u,v \in V$.
Chepoi et al. showed in \cite{ChDrEsHaVa08} that $\delta$-hyperbolic graphs have good approximating trees that can be computed in linear ($O(m)$) time.

\begin{theorem}[\cite{ChDrEsHaVa08}]\label{thm:gromovtree}
Let $G = (V,E) $ be a $\delta$-hyperbolic graph, and let $\Lambda_n = \lceil 4 + 3\delta + 2\delta \log_2 n \rceil$.
There exists a $\Lambda_n$-approximating tree $T = (V,F)$ of $G$.
Furthermore $T$ can be computed from $G$ in time $O(m)$.
\end{theorem}

Fredrickson \cite{Fr91} showed that $p$-centres can be solved in linear time on trees.

\begin{theorem}[\cite{Fr91}]\label{thm:linearpcentre}
Let $T$ be a tree and $p$ an integer.
There exists an algorithm to solve $p$-centres exactly on $T$ in time $O(n)$.
\end{theorem}

Shier proved in \cite{SH77} (see Theorem \ref{cor:shier}) that in trees, the $p$-radius is always half the $p+1$-diameter. 
In \cite{ChDa81} Chandrasekaran and Daughety gave an algorithm to find the $p$-diameter (and an optimal packing of size $p$) in a tree in time $O(n^2\log n)$.
Their technique involves a binary search for $\lambda_p$ through repeated application of a subroutine which, when given a half-integer $\lambda$, produces a maximum number of points which are pairwise at distance $>2\lambda$.
More precisely,

\begin{theorem}[\cite{ChDa81}]\label{thm:dispersion}
Let $T = (V,E)$ be a tree and let $2\lambda$ be an integer.
There exists an algorithm which, in time $O(n\log n)$, produces a set $W\subseteq V$ of maximum size such that $d(u,v) \added{\geq} 2\lambda$ for each $u\neq v \in V$.
\end{theorem}

Combining Theorem \ref{thm:dispersion} with Theorem \ref{cor:shier} and Theorem \ref{thm:linearpcentre} we obtain an $O(n\log n)$ algorithm to find an optimal packing of size $p$ in a tree.

Now, suppose that $G$ is $\delta$-hyperbolic and $T$ is a $\Lambda_n$-approximating tree for $G$.
By definition of an approximating tree, for every $u,v \in V$ we have $|d_T(u,v) - d_G(u,v)| \leq \Lambda_n$.
It follows that $|d_{p+1}(T) - d_{p+1}(G)| \leq \Lambda_n$.
Thus we obtain Algorithm \ref{alg:initial} and Theorem \ref{cor:dispersion} that yields our initial $(p+1)$-packing.

At this point, the reader may be wondering why we go to the trouble of Algorithm \ref{alg:initial} to obtain our initial $(p+1)$-packing. Indeed, one could start with any packing at the beginning of Algorithm \ref{alg:optopt}, and repeat rounds of Subroutines \ref{alg:potential} and \ref{alg:vertex} until a packing satisfying Properties (a), (b) and (c) is found. However, as we have just seen, the number of times we may need to repeat the rounds is upper bounded by the difference between the dispersion of our initial set and the optimal dispersion $d_{p+1}$. 
When the initial set is chosen using Algorithm \ref{alg:initial}, this difference is $O(\delta\log n)$,
whereas trying to save time choosing the initial set (say, by choosing it arbitrarily) may result in an additional linear factor in the complexity bound. Applying the greedy $2$-approximation algorithm of Ravi, Rosenkrantz and Tayi mentioned in Section \ref{sec:palgs} adds a factor of $d_{p+1}$, which may also be linear. However, the practitioner may wish to experiment.

\section{Empirical results}
\label{sec:exp}

\begin{table}[!htb]
\centering
\caption{Networks analyzed}
\label{tab:datasummary}
\begin{tabular}{|l|c|c|c|c|c|c|}
\hline
Network         & Type &\multicolumn{1}{c|}{$|V|$} & $|E|$ & diameter & \multicolumn{1}{c|}{radius} & \multicolumn{1}{c|}{$\delta_{4-point}$} \\ \hline \hline
sprintlink-1239 &  Rocketfuel ISP network  &         8341        & 14025   &    13    &         7      &          3        \\ \hline
p2p-gnutella25  &  peer-to-peer network   &    22663          & 54693 &     11   &       7       &       3         \\ \hline
sn-medium       &  social network   &       26567        & 226566  &    14    &        7      &          4        \\ \hline
web-stanford    &  web network  &   255265   & 1941926  &    164    &      82        &       1.5 (est.)           \\ \hline
\end{tabular}
\end{table}

We have implemented the algorithms from Theorem \ref{thm:generalp} ($p\geq 3$) and Theorem \ref{thm:smallp} ($p\leq 2$).
For comparison, we have also implemented the algorithms of Chepoi et al. (Ch.) \cite{ChDrEsHaVa08} ($p=1$) and Chepoi-Estellon (C-E) \cite{ChEs07} ($p\geq 2$).
We also compared Theorems \ref{thm:generalp} and \ref{thm:smallp} to the following simple algorithm:
Compute a distance approximating tree as in Theorem \ref{thm:gromovtree} and return an exact solution to $p$-centres on $T$.

We ran the algorithms on four graphs extracted from real networks arising from different types of data.
All graphs are simple and have unit edge lengths and each has a small hyperbolicity constant.
Table \ref{tab:datasummary} briefly summarizes the networks we analyzed; more information about the data can be found in \cite{KNS}.
\footnote{The graphs p2p-gnutella25 and web-stanford are available publicly as part of the Stanford Large Network Dataset Collection.
The sn-medium graph is extracted from the social network Facebook, and the sprintlink-1239 graph is an IP-layer network from the Rocketfuel ISP.}
In the case of the web-stanford graph, we have only an estimate of $\delta_{4-point}$ obtained by sampling since the graph is quite large.
Table \ref{tab:results} contains a comparison of the estimated $p$-radius $r_p$ of the three algorithms.
We have run only our algorithm from Theorem \ref{sec:generalp} on the largest network (web-stanford), since the running time of C-E is infeasible on a graph of this size.

Our experiments indicate that, as far as accuracy goes, our algorithm performs similarly to that of Chepoi-Estellon despite the larger theoretical upper bound on the error. 
In many cases, in fact, our estimate is better than that one.
The $p$-radius estimated by the algorithm in Theorem \ref{sec:generalp} is always within $1$ of their estimate in our trials.
Combined with the significant improvement in running time, this makes our algorithm a preferable choice for solving $p$-centres in practice.
For comparison, our implementation of our algorithm terminated in under two seconds on the sn-medium graph, while C-E took about one minute.

While the tree-approximation heuristic is simple, and runs in quasilinear time $O(m+n)$, the approximation guarantee is only as good as the distance approximation of $T$, hence the additive error could up to $O(\delta \log n)$.
However, our experiments show that it seems to perform well in practice and may be a good choice of heuristic in some applications.

\begin{table}[!htb]
\centering
\caption{Comparison of estimates of the $p$-radius.}
\label{tab:results}
\begin{tabular}{l|c|c|c|c|c|c|
>{\columncolor[HTML]{EFEFEF}}c |
>{\columncolor[HTML]{EFEFEF}}c |
>{\columncolor[HTML]{EFEFEF}}c |c|c|c|}
\hhline{~------------} 
 & \multicolumn{3}{c|}{\cellcolor[HTML]{EFEFEF}sprintlink-1239} & \multicolumn{3}{c|}{p2p-gnutella25} & \multicolumn{3}{c|}{\cellcolor[HTML]{EFEFEF}sn-medium} & \multicolumn{3}{c|}{\cellcolor[HTML]{FFFFFF}web-stanford} \\ \hhline{~------------}
 & \multicolumn{1}{l|}{\cellcolor[HTML]{EFEFEF}Thm 9} & \multicolumn{1}{l|}{\cellcolor[HTML]{EFEFEF}Ch.} & \multicolumn{1}{l|}{\cellcolor[HTML]{EFEFEF}Tree} & \multicolumn{1}{l|}{Thm 9} & \multicolumn{1}{l|}{Ch.} & \multicolumn{1}{l|}{Tree} & \multicolumn{1}{l|}{\cellcolor[HTML]{EFEFEF}Thm 9} & \multicolumn{1}{l|}{\cellcolor[HTML]{EFEFEF}Ch.} & \multicolumn{1}{l|}{\cellcolor[HTML]{EFEFEF}Tree} & \cellcolor[HTML]{FFFFFF}Thm 9 & \cellcolor[HTML]{FFFFFF}Ch. & \cellcolor[HTML]{FFFFFF}Tree \\ \hline
\multicolumn{1}{|l|}{$p=1$} & 7 & 7 & 8 & 8 & 8 & 7 & 7 & 7 & 8 & 82 &  &  \\ \hline
\multicolumn{1}{|l|}{$p=2$} & 7 & 7 & 7 & 8 & 8 & 7 & 7 & 8 & 8 & 59 &  &  \\ \hline
 & \cellcolor[HTML]{EFEFEF}Thm 8 & \cellcolor[HTML]{EFEFEF}C-E & \cellcolor[HTML]{EFEFEF}Tree & Thm 8 & C-E & Tree & Thm 8 & C-E & Tree & \cellcolor[HTML]{FFFFFF}Thm 8 & \cellcolor[HTML]{FFFFFF}C-E & \cellcolor[HTML]{FFFFFF}Tree \\ \hline
\multicolumn{1}{|l|}{$p=3$} & \cellcolor[HTML]{EFEFEF}5 & \cellcolor[HTML]{EFEFEF}6 & \cellcolor[HTML]{EFEFEF}6 & 7 & 7 & 7 & 7 & 7 & 8 & \cellcolor[HTML]{FFFFFF}47 & \cellcolor[HTML]{FFFFFF} & \cellcolor[HTML]{FFFFFF} \\ \hline
\multicolumn{1}{|l|}{$p=4$} & \cellcolor[HTML]{EFEFEF}5 & \cellcolor[HTML]{EFEFEF}6 & \cellcolor[HTML]{EFEFEF}6 & 7 & 7 & 7 & 6 & 7 & 8 & \cellcolor[HTML]{FFFFFF}46 & \cellcolor[HTML]{FFFFFF} & \cellcolor[HTML]{FFFFFF} \\ \hline
\multicolumn{1}{|l|}{$p=5$} & \cellcolor[HTML]{EFEFEF}4 & \cellcolor[HTML]{EFEFEF}5 & \cellcolor[HTML]{EFEFEF}6 & 7 & 6 & 7 & 6 & 7 & 8 & \cellcolor[HTML]{FFFFFF}44 & \cellcolor[HTML]{FFFFFF} & \cellcolor[HTML]{FFFFFF} \\ \hline
\multicolumn{1}{|l|}{$p=6$} & \cellcolor[HTML]{EFEFEF}4 & \cellcolor[HTML]{EFEFEF}5 & \cellcolor[HTML]{EFEFEF}6 & 7 & 6 & 7 & 6 & 6 & 8 & \cellcolor[HTML]{FFFFFF}44 & \cellcolor[HTML]{FFFFFF} & \cellcolor[HTML]{FFFFFF} \\ \hline
\multicolumn{1}{|l|}{$p=7$} & \cellcolor[HTML]{EFEFEF}4 & \cellcolor[HTML]{EFEFEF}5 & \cellcolor[HTML]{EFEFEF}5 & 6 & 6 & 7 & 6 & 6 & 8 & \cellcolor[HTML]{FFFFFF}44 & \cellcolor[HTML]{FFFFFF} & \cellcolor[HTML]{FFFFFF} \\ \hline
\multicolumn{1}{|l|}{$p=8$} & \cellcolor[HTML]{EFEFEF}4 & \cellcolor[HTML]{EFEFEF}5 & \cellcolor[HTML]{EFEFEF}5 & 6 & 6 & 7 & 6 & 6 & 8 & \cellcolor[HTML]{FFFFFF}38 & \cellcolor[HTML]{FFFFFF} & \cellcolor[HTML]{FFFFFF} \\ \hline
\multicolumn{1}{|l|}{$p=9$} & \cellcolor[HTML]{EFEFEF}4 & \cellcolor[HTML]{EFEFEF}5 & \cellcolor[HTML]{EFEFEF}5 & 6 & 6 & 7 & 6 & 6 & 7 & \cellcolor[HTML]{FFFFFF}29 & \cellcolor[HTML]{FFFFFF} & \cellcolor[HTML]{FFFFFF} \\ \hline
\multicolumn{1}{|l|}{$p=10$} & \cellcolor[HTML]{EFEFEF}4 & \cellcolor[HTML]{EFEFEF}5 & \cellcolor[HTML]{EFEFEF}5 & 6 & 6 & 7 & 6 & 6 & 7 & \cellcolor[HTML]{FFFFFF}29 & \cellcolor[HTML]{FFFFFF} & \cellcolor[HTML]{FFFFFF} \\ \hline
\multicolumn{1}{|l|}{$p=11$} & \cellcolor[HTML]{EFEFEF}4 & \cellcolor[HTML]{EFEFEF}5 & \cellcolor[HTML]{EFEFEF}5 & 6 & 6 & 7 & 6 & 6 & 7 & \cellcolor[HTML]{FFFFFF}27 & \cellcolor[HTML]{FFFFFF} & \cellcolor[HTML]{FFFFFF} \\ \hline
\multicolumn{1}{|l|}{$p=12$} & \cellcolor[HTML]{EFEFEF}4 & \cellcolor[HTML]{EFEFEF}4 & \cellcolor[HTML]{EFEFEF}5 & 6 & 6 & 7 & 5 & 6 & 7 & \cellcolor[HTML]{FFFFFF}23 & \cellcolor[HTML]{FFFFFF} & \cellcolor[HTML]{FFFFFF} \\ \hline
\multicolumn{1}{|l|}{$p=13$} & \cellcolor[HTML]{EFEFEF}4 & \cellcolor[HTML]{EFEFEF}4 & \cellcolor[HTML]{EFEFEF}5 & 6 & 5 & 7 & 5 & 6 & 7 & \cellcolor[HTML]{FFFFFF}23 & \cellcolor[HTML]{FFFFFF} & \cellcolor[HTML]{FFFFFF} \\ \hline
\multicolumn{1}{|l|}{$p=14$} & \cellcolor[HTML]{EFEFEF}4 & \cellcolor[HTML]{EFEFEF}4 & \cellcolor[HTML]{EFEFEF}5 & 6 & 5 & 7 & 5 & 6 & 7 & \cellcolor[HTML]{FFFFFF}23 & \cellcolor[HTML]{FFFFFF} & \cellcolor[HTML]{FFFFFF} \\ \hline
\multicolumn{1}{|l|}{$p=15$} & \cellcolor[HTML]{EFEFEF}4 & \cellcolor[HTML]{EFEFEF}4 & \cellcolor[HTML]{EFEFEF}5 & 6 & 5 & 7 & 5 & 6 & 7 & \cellcolor[HTML]{FFFFFF}22 & \cellcolor[HTML]{FFFFFF} & \cellcolor[HTML]{FFFFFF} \\ \hline
\multicolumn{1}{|l|}{$p=16$} & \cellcolor[HTML]{EFEFEF}4 & \cellcolor[HTML]{EFEFEF}4 & \cellcolor[HTML]{EFEFEF}5 & 6 & 5 & 7 & 5 & 6 & 7 & \cellcolor[HTML]{FFFFFF}21 & \cellcolor[HTML]{FFFFFF} & \cellcolor[HTML]{FFFFFF} \\ \hline
\multicolumn{1}{|l|}{$p=17$} & \cellcolor[HTML]{EFEFEF}4 & \cellcolor[HTML]{EFEFEF}4 & \cellcolor[HTML]{EFEFEF}5 & 6 & 5 & 7 & 5 & 6 & 7 & \cellcolor[HTML]{FFFFFF}20 & \cellcolor[HTML]{FFFFFF} & \cellcolor[HTML]{FFFFFF} \\ \hline
\multicolumn{1}{|l|}{$p=18$} & \cellcolor[HTML]{EFEFEF}4 & \cellcolor[HTML]{EFEFEF}4 & \cellcolor[HTML]{EFEFEF}5 & 5 & 5 & 7 & 5 & 6 & 7 & \cellcolor[HTML]{FFFFFF}19 & \cellcolor[HTML]{FFFFFF} & \cellcolor[HTML]{FFFFFF} \\ \hline
\multicolumn{1}{|l|}{$p=19$} & \cellcolor[HTML]{EFEFEF}4 & \cellcolor[HTML]{EFEFEF}4 & \cellcolor[HTML]{EFEFEF}5 & 5 & 5 & 7 & 5 & 6 & 7 & \cellcolor[HTML]{FFFFFF}16 & \cellcolor[HTML]{FFFFFF} & \cellcolor[HTML]{FFFFFF} \\ \hline
\multicolumn{1}{|l|}{$p=20$} & \cellcolor[HTML]{EFEFEF}4 & \cellcolor[HTML]{EFEFEF}4 & \cellcolor[HTML]{EFEFEF}5 & 5 & 5 & 7 & 5 & 6 & 7 & \cellcolor[HTML]{FFFFFF}17 & \cellcolor[HTML]{FFFFFF} & \cellcolor[HTML]{FFFFFF} \\ \hline
\end{tabular}
\end{table}

\bibliography{hyperbolic}{}
\bibliographystyle{plain}

\end{document}